\documentclass[runningheads,a4paper]{llncs}

\usepackage{graphicx}

\usepackage{multirow}
\usepackage{booktabs}

\usepackage{amsmath}
\usepackage{amssymb}

\usepackage{url}
\newcommand{\keywords}[1]{\par\addvspace\baselineskip
\noindent\keywordname\enspace\ignorespaces#1}

\newcommand{\A}{\mathcal{A}}
\newcommand{\N}{\mathbb{N}}

\newcommand{\thmA}{{C_0}}
\newcommand{\thmB}{{C_1}}
\newcommand{\thmC}{{C_2}}

\newcommand{\pop}{{\rm {\bf pop}}}
\newcommand{\weight}{{\rm {\bf weight}}}
\newcommand{\midp}{{\rm {\bf midpoint}}}
\newcommand{\depth}{{\rm {\bf depth}}}
\newcommand{\densify}{{\rm {\bf densify}}}
\newcommand{\starts}{{\rm {\bf start}}}
\newcommand{\sends}{{\rm {\bf end}}}

\renewenvironment{proof}{\noindent{\it Proof. }} {{\qed}}
\newenvironment{proofof}[1]{\noindent{\it Proof of #1. }} {{\qed}}

\begin{document}

\mainmatter  % start of an individual contribution

\title{On Online Labeling with Polynomially Many Labels}

\author{Martin Babka\inst{1}%
\and Jan Bul\'anek\inst{1,2}%
\and Vladim\'ir \v{C}un\'at\inst{1}
\thanks{The first three authors gratefully acknowledge a support by the Charles University Grant Agency (grant No. 265 111 and 344 711) and SVV project no. 265 314.}%
\and Michal Kouck\'y\inst{2,3}
\thanks{Currently a visiting associate professor at Aarhus University, partially supported by the Sino-Danish Center CTIC (funded under the grant 61061130540). Supported in part by GA \v{C}R P202/10/0854, grant IAA100190902 of GA AV \v{C}R, Center of Excellence CE-ITI (P202/12/G061 of GA \v{C}R) and RVO: 67985840.}
\and Michael Saks\inst{4}
\thanks{The work of this author was done while on sabbatical at Princeton University and
was also supported in part by NSF under  grant CCF-0832787.}
}

\institute{Faculty of Mathematics and Physics, Charles University, Prague\\
\and Institute of Mathematics, Academy of Sciences, Prague\\
\and Department of Computer Science, Aarhus University\\
\and Department of Mathematics, Rutgers University\\\vspace{2mm} 
\path| {babkys,vcunat}@gmail, {bulda,koucky}@math.cas.cz, saks@math.rutgers.edu |}

\authorrunning{On Online Labeling with Polynomially Many Labels}

% \institute{Springer-Verlag, Computer Science Editorial,\\
% Tiergartenstr. 17, 69121 Heidelberg, Germany\\
% \mailsa\\
% \mailsb\\
% \mailsc\\
% \url{http://www.springer.com/lncs}}


\toctitle{Lecture Notes in Computer Science}
\tocauthor{Authors' Instructions}
\maketitle

\begin{abstract}

In the online labeling problem with parameters $n$ and $m$
we are presented with a sequence of $n$ {\em keys} from a totally ordered universe $U$
and must assign each arriving key a label from the label set $\{1,2,\dots,m\}$ so that the order of labels (strictly) respects the ordering on $U$.
As new keys arrive it may be necessary to change the labels of some items; such changes may be done at any time at unit cost
for each change.    
The goal is to minimize
the total cost.
An alternative formulation of this problem is
the \emph{file maintenance problem}, in which the items, instead of being labeled, 
are maintained in sorted order in an  array of length $m$, and we pay unit cost for moving an item. 

%%MIKE 6-27-12
For the case $m=cn$ for constant $c>1$, 
there are known algorithms that use at most $O(n \log(n)^2)$ relabelings in total \cite{Itaietal}, 
and it was shown recently that this is asymptotically optimal \cite{BKS}. For the case of $m=\theta(n^C)$ for $C>1$,
algorithms are known that use $O(n \log n)$ relabelings. A matching lower bound was claimed in \cite{DSZ04}.  That
proof involved two distinct steps: a lower bound for a problem they call {\em prefix bucketing}
and a reduction from  prefix bucketing to online labeling. The reduction seems to be incorrect, leaving a (seemingly significant)
gap in the proof.
In this paper we close the gap by presenting a correct reduction to prefix bucketing.  Furthermore
we give a simplified and improved analysis of the prefix bucketing lower bound.  This improvement allows us to extend the
lower bounds for online labeling to the case where the number $m$ of labels is superpolynomial in $n$.
In particular, for superpolynomial $m$ we get an asymptotically optimal lower  bound 
$\Omega((n \log n) / (\log \log m - \log\log n))$.

\end{abstract}

\section{Introduction}

In the online labeling problem with parameters $n,m,r$, 
we are presented with a sequence of $n$ {\em keys} from a totally ordered universe $U$ of size $r$
and must assign each arriving key a label from the label set $\{1,2,\dots,m\}$ so that the order of labels (strictly) respects the ordering on $U$.
As new keys arrive it may be necessary to change the labels of some items; such changes may be done at any time at unit cost
for each change.    
The goal is to minimize
the total cost.
An alternative formulation of this problem is
the \emph{file maintenance problem}, in which the items, instead of being labeled, 
are maintained in sorted order in an  array of length $m$, and we pay unit cost for moving an item. 

The problem, which was introduced by Itai, Konheim and Rodeh \cite{Itaietal}, 
is natural and intuitively appealing, and has had applications to the design of data structures (see for example
the discussion in \cite{DSZ04}, and the more recent work on cache-oblivious data structures 
\cite{BenderetalB-Tree,Brodaletal,BDIW}).  A connection between this problem and distributed resource allocation
was recently shown by Emek and Korman \cite{EK11}.

The parameter $m$, the {\em label space} must be at least the number of items $n$ or else no valid labeling is possible.  
There are two natural range of parameters that have received the most attention. In the case of {\em linearly many labels} we have
$m=cn$ for some $c>1$, and in the case of {\em polynomially many labels} we have $m=\theta(n^{C})$ for some constant $C>1$.  
The problem is trivial if the universe $U$ is a set of size at most $m$, since then we can simply fix an order preserving
bijection from $U$ to $\{1,\ldots,m\}$ in advance.  In this paper we will usually restrict attention to the case that  $U$
is a totally ordered set of size at least exponential in $n$ (as is typical in the literature).  

Itai et al. \cite{Itaietal} gave an algorithm for the case of linearly many labels having worst case total cost 
$O(n \log(n)^2)$.  Improvements and simplifications were given by Willard \cite{Willard} and Bender et al. \cite{Benderetal}. 
In the special case that $m=n$, algorithms with cost $O(\log(n)^3)$ per item were given
\cite{Zhang,BirdSadnicki}.  It is also
well known that the  algorithm of Itai et al. can be adapted to give total cost $O(n \log(n))$ in the case
of polynomially many labels.  
All of these algorithms make no restriction on the size of
universe $U$ of keys.

For lower bounds, a subset of the present authors 
recently proved \cite{BKS} a tight $\Omega(n \log(n)^2)$ lower bound for the case of linearly many labels and tight bound $\Omega(n \log(n)^3)$ for
the case $m=n$. These bounds hold even when the size of the universe $U$ is only a constant multiple of $m$.  The bound remains
non-trivial (superlinear in $n$) for 
%%MIKE 6-27-2012
$m  = O(n \log(n)^{2-\varepsilon})$ but
becomes trivial  for $m \in \Omega(n \log(n)^2)$. 

For the case of polynomially many labels,
Dietz at al. \cite{DSZ04} (also in \cite{Zhang}) claim a matching lower bound for the $O(n \log(n))$ upper bound.
Their result consists of two parts; a lower bound for a problem they call {\em prefix bucketing}
and a reduction from  prefix bucketing to online labeling.  However, the argument giving the reduction seems to be incorrect, 
and we recently raised our concerns with one of the authors (Seiferas), who agrees that
there is a gap in the proof. 

This paper makes the following contributions:

\begin{itemize}
\item We provide a correct reduction from prefix bucketing to  online labeling, which closes the gap in the lower
bound proof for online labeling for the case of polynomially many labels.
\item We provide a simpler and more precise lower bound for prefix bucketing which allows us to extend the lower bounds
for online labeling to the case where the label size is as large as $2^{n^\epsilon}$.
Specifically we prove a lower bound of $\Omega((n \log n) / (\log \log m - \log \log n))$ that is valid
for $m$ between $n^{1+\varepsilon}$ and $2^{n^{\varepsilon}}$.   Note that for
polynomially many labels this reduces to $\Omega(n \log(n))$. 
\end{itemize}

We remark that, unlike the bounds of \cite{BKS} for the case of linearly many labels, our lower bound proof requires that the universe
$U$ is at least exponential in $n$. It is an interesting question whether one could design
a better online labeling algorithm for $U$ of size say $m \log n$. We summarize known results in Table \ref{tab-1}. All the results
are for deterministic algorithms. There are no known results for randomized algorithms except for what is implied by the deterministic case.

\begin{table}
\centering
\caption{Summary of known bounds for the online labeling problem.} \label{tab-1}
\begin{tabular}{l lr lr}
\toprule
Array size ($m$)            & Lower bound                                                      &                             & Upper bound                        & \\ \toprule
$m = n$                     & $\Omega\!\left(n \log(n)^3\right)$                               & \cite{BKS}\hspace{.5cm}     & $O\!\left(n \log(n)^3\right)$      & \cite{Zhang} \\ \midrule
$m = \Theta(n), m > n\,\,\,\,\,\,\,$      & $\Omega\!\left(n \log(n)^2\right)$                 & \cite{BKS}\hspace{.5cm}     & $O\!\left(n \log(n)^2\right)$      & \cite{Itaietal} \\ \midrule
$m = n ^ {1 + o(1)}$        & $\Omega\!\left(\frac{n\log (n)^2}{\log m - \log n}\right)$       & \cite{manuscript}\hspace{.5cm}  & $O\!\left(\frac{n \log (n)^2}{\log m - \log n}\right)$ & \cite{Itaietal} \\ \midrule
$m = n ^ {1 + \Theta(1)}$   & $\Omega\!\left(n\log (n)\right)$                                 & [this paper]\hspace{.5cm}   & $O\!\left(n \log (n)\right)$       & \cite{Itaietal} \\ \midrule
$m = n ^ {\Omega(\log(n))}$ & $\Omega\!\left(\frac{n \log n}{\log \log m}\right)$              & [this paper]\hspace{.5cm}   & $O\!\left(\frac{n \log n}{\log \log m}\right)$       & \cite{BKS} \\ \bottomrule
\end{tabular}
\label{table:table_bounds}
\end{table}

Our proof follows the high level structure of the proof from \cite{DSZ04}.  In the remainder of the introduction we sketch the two
parts, and relate our proof to the one in \cite{DSZ04}.

\subsection{Reducing Online Labeling to Prefix Bucketing}

Dietz et al. \cite{DSZ04} sketched a reduction from online labeling to prefix bucketing.  In their reduction they describe an adversary
for the labeling problem.  They show that given any algorithm for online labeling, the behavior of the algorithm
against the adversary can be used to construct a strategy for prefix bucketing.  If one can show that the cost of the derived bucketing strategy is no more than a constant times  the cost paid by the algorithm for relabelings then a lower bound on bucketing
will give a similar lower bound on the cost of any relabeling algorithm.  Unfortunately, their proof sketch does not show this.
In particular, a single relabeling step may correspond to a bucketing step whose cost is $\Omega(\log(n))$, and this
undermines the reduction. 
This may happen when inserting $\Theta(\log n)$ items into an empty segment of size $n^\epsilon$ without triggering any relabelings.
We construct a different adversary for which one gets the needed correspondence
between relabeling cost and bucketing steps.

Our goal in constructing an adversary is to force any online algorithm to perform many relabelings during insertion on $n$ keys. 
The adversary is described
in detail in Section \ref{s-olp} here we provide a high level description.

The adversary starts by inserting the minimal and maximal element of $U$, i.e. $1$ and $r$, and five other keys uniformly spread in $U$.
From then on the adversary will always pick a suitable pair of existing consecutive keys and the next inserted key
will be the average of the pair. Provided that $r>2^n$ there will always be an unused key between the pair.

It is illuminating to think of the problem in terms of the file maintenance problem
mentioned earlier.  In this reformulation we associate to the label space $[m]$ an array indexed by $\{1,\ldots,m\}$ and think of
a key labeled by $j$ as being stored in location $j$. Intuitively, 
the adversary wants to choose two consecutive keys that appear in a crowded region of this array.  By doing this repeatedly,
the adversary hopes to force
the algorithm to move many keys within the array (which corresponds to relabeling them). 
The problem is to make precise the notion of ``crowdedness''.  Crowding within the array occurs at different scales (so a small crowded region may lie inside
a large uncrowded region)  and we need to find a pair of consecutive keys with the property that all regions containing the pair are
somewhat crowded.

With the array picture in mind,
we  call an interval of labels a \emph{segment}, and say that a label is {\em occupied} if there is a key assigned to it. 
The \emph{density} of a segment is the fraction of occupied labels.   

In \cite{DSZ04}, the authors show that there is always a \emph{dense point}, which is
a point  with the property that every segment containing it
has density at least half the overall density of the label space.  They use this as the basis for their
adversary, but this adversary does not seem to be adequate.

%%MIKE 6-27-12
We design a different adversary (which is related to the adversary constructed in \cite{BKS} to handle
the case of linearly many labels).
The adversary maintains a sequence (\emph{hierarchy}) of nested segments.  Each successive segment in the hierarchy
has size at most half the previous segment, and its density  is within a constant factor of the density of the previous segment.
The hierarchy ends with a segment having between 2 and 7 keys.
 The key chosen to be inserted next is the average (rounded down) of the two ``middle'' keys of the smallest segment.

For the first 8 insertions, the hierarchy consists only of the single segment $[m]$.
After each subsequent insertion,
the algorithm $\A$ specifies the label of the next key and (possibly) relabels some keys. The adversary then updates
its hierarchy.  For the hierarchy prior to the insertion, define the critical segment to be the smallest segment
of the hierarchy that contains the label assigned to the inserted key and the old and new labels
of all keys that were relabeled.  The new hierarchy agrees with the previous hierarchy up to and including
the critical segment.   Beginning from the critical segment the hierarchy is extended as follows.
Having chosen segment $S$ for the hierarchy, define its 
{\em left buffer} 
to be the smallest subsegment of $S$ that starts at the minimum label of $S$ and
includes at least 1/8 of the occupied labels of $S$, and its {\em right buffer} to be the smallest subsegment that ends at the maximum label of $S$
and includes at least 1/8 of the occupied keys of $S$.
Let $S'$ be the segment obtained from $S$ by
deleting the left and right buffers. The successor segment of $S$ in the hierarchy is
the shortest subsegment of $S'$ that contains exactly half (rounded down) of the occupied labels of $S'$.
The hierarchy ends when we reach a segment with at most seven occupied labels; such a segment necessarily has at least two occupied
labels. 

It remains to prove that the algorithm will make lot of relabels on the sequence
of keys produced by the adversary. 
For that proof we need a game introduced by Dietz et al. \cite{DSZ04} that they
call prefix bucketing.

%%MIKE 6-27-2012
A \emph{prefix bucketing of $n$ items into $k$  buckets} (numbered 1 to $k$) is a one player game consisting of $n$ steps. 
At the beginning of the game
all the buckets are empty. In each step a new item arrives and the player selects an index $p \in [k]$.
The new item as well as all items in buckets $p+1,\ldots,k$ are moved into bucket $p$ at a cost equal to the
total number of items in bucket $p$ after the move.  The goal is to minimize the total cost of $n$ steps 
of the game. 
Notice that suffix bucketing would be a more appropriate name for our variant of the game, however we keep the original name as the games are equivalent.

We will show that if $\A$ is any algorithm for online labeling and we run $\A$ against our adversary then
the behavior of $\A$ corresponds to a prefix bucketing of $n$ items into $k=\lceil \log m \rceil$ buckets.
The total cost of the prefix bucketing will be within a constant factor of  the total number of relabelings
performed by the online labeling algorithm. Hence, a lower bound on the cost of a prefix bucketing of $n$ items into $k$ 
buckets will imply a lower bound on the cost of the algorithm against our adversary. 

Given the execution of $\A$ against our adversary we create 
the following prefix bucketing.  We maintain a \emph{level} for each key inserted by the adversary; one invariant
these levels will satisfy is that for each segment in the hierarchy, the level of the keys inside
the segment are at least the depth of the segment in the hierarchy.
The level of a newly inserted key is initially $k$. After inserting the key, the algorithm
does its relabeling, which determines the critical segment and the critical level.
All keys in the current critical segment whose level exceeds the critical level have their levels
reset to the current critical level.   

The assignment of levels to keys corresponds to a bucketing strategy, where the level of a key is the bucket
it belongs to.   Hence, if $p$ is the critical level, all the keys from 
buckets $p,p+1,\dots,k$ will be merged into the bucket $p$. 

We need to show that 
the cost of the merge operation corresponds to the number of relabelings done by the online algorithm
at a given time step.  For this we  make the assumption (which can be shown to hold without loss of generality) that the algorithm is \emph{lazy}, which means
that at each time step the set of keys that are relabeled is a contiguous block of keys that includes
the newly inserted keys.    The cost of the bucketing merge step is at most the number of keys in
the critical segment.  One can argue that for each successor segment of the critical segment, either all labels in its left buffer
or all labels in its right buffer were reassigned, and the total number of such keys is a constant fraction of the keys
in the critical segment.

\subsection{An Improved Analysis of Bucketing}

It then remains to give a lower bound
for the cost of prefix bucketing.  This was previously given by Dietz et al. \cite{DSZ04} for $k\in \Theta(\log n)$. 
We give a different and simpler proof that gives asymptotically optimal bound for $k$ between $\log n$ and $O(n^\epsilon)$.
We define a family of trees called $k$-admissible trees and show that the cost
of bucketing  for $n$ and $k$, is between $dn/2$ and $dn$ where $d$ is the minimum depth of a $k$-admissible tree
on $n$ vertices.  We further show that the minimum depth of a $k$-admissible tree on $n$ vertices
is equal $g_k(n)$ which is defined to be the smallest $d$ such that $\binom{k+d-1}{k} \geq n$.  
This gives a  characterization of the optimal cost of prefix bucketing (within a factor of 2).  When we apply
this characterization we need to use estimates of $g_k(n)$ in terms of more familiar functions (Lemma \ref{lm:lower_bound_d}), 
and there is some loss in these estimates.

\section{The Online Labeling Problem}\label{s-olp}

In this paper, interval notation is used for sets of consecutive integers, e.g.,
$[a,b]$ is the set $\{k \in \mathbb{Z}:a \leq k \leq b\}$. Let $m$ and $r\ge 1$ be integers. We assume without loss of generality $U=[r]$.
An online labeling algorithm $\A$ with range $m$ is an algorithm that on input sequence $y^1,y^2,\dots,y^t$ of distinct elements
from $U$ gives
an \emph{allocation $f : \{y^1,y^2,\dots,y^t\} \rightarrow [m]$} that respects the natural ordering of $y^1,\dots,y^t$ 
that is for any $x,y\in \{y^1,y^2,\dots,y^t\}$, $f(x) < f(y)$ if and only if $x < y$. We refer to $y^1,y^2,\dots,y^t$
as \emph{keys}. The trace of $\A$ on a sequence $y^1,y^2,\dots,y^n\in U$ is the sequence $f^0,f^1,f^2,\dots,f^n$ of functions
such that $f^0$ is the empty mapping and for $i=1,\dots,n$, $f^i$ is the output of $\A$ on $y^1,y^2,\dots,y^i$. 
For the trace $f^0,f^1,f^2,\dots,f^n$ and $t=1,\dots,n$, we say that $\A$ \emph{relocates $y\in \{y^1,y^2,\dots,y^t\}$ at time $t$} 
if $f^{t-1}(y)\not=f^t(y)$. So each $y^t$ is relocated at time $t$. 
For the trace $f^0,f^1,f^2,\dots,f^n$ and $t=1,\dots,n$, $Rel^t$ denotes the set of relocated keys at step $t$.
The cost of $\A$ incurred on $y^1,y^2,\dots,y^n\in U$ is $\chi_\A(y_1,\dots,y_n)=\sum_{i=0}^n |Rel^i|$ where $Rel$ is 
measured with respect to the trace of $\A$ on $y^1,y^2,\dots,y^n$. The maximum cost
$\chi_\A(y^1,\dots,y^n)$ over all sequences $y^1,\dots,y^n$ is denoted $\chi_\A(n)$. We write $\chi_m(n)$ for the smallest cost $\chi_\A(n)$ that 
can be achieved by any algorithm $\A$ with range $m$.
%% Michal: ``algorithm with range m'' is a technical term, so ``label set $[m]$.'' is not good.

\subsection{The Main Theorem}
\label{subsec:main theorem}
In this section, we state our lower bound results for  $\chi_m(n)$.   

\begin{theorem}
\label{thm:main}
There are positive constants $\thmA, \thmB$ and $\thmC$ so that the following holds. Let $m,n$ be integers satisfying
$\thmA \leq n \le m \le 2^{n^\thmB}$. Let the size of $U$ be more than $2^{n+4}$. 
Then $\chi_m(n) \geq \thmC \cdot \frac{n \log n}{3 + \log \log m - \log \log n}$.
\end{theorem}

To prove the theorem for given algorithm $\A$ we will adversarially construct a sequence $y^1,y^2,\dots,y^t$ of keys that will cause
the algorithm to incur the desired cost. In the next section we will design the adversary.

\subsection{Adversary Construction}

Any interval $[a,b]\subseteq [m]$ is called a \emph{segment}. Fix $n,m,r>1$ such that $m\ge n$ and $r>2^{n+4}$.
Fix some online labeling algorithm $\A$ with range $m$. 
To pick the sequence of keys $y^1,\dots,y^n$, the adversary will maintain a sequence of 
nested segments $S^t_{\depth(t)} \subseteq \cdots \subseteq S^t_2 \subseteq S^t_1 = [m]$, updating them after each time step $t$. 
The adversary will choose the next element $y^{t}$ to fall between the keys in
the smallest interval $S^t_{\depth(t)}$. In what follows, $f^t$ is the allocation of $y^1,\dots,y^t$ by the algorithm $\A$.

The \emph{population of a segment $S$ at time $t$} is $\pop^t(S) = (f^t)^{-1}(S)$ and the 
\emph{weight of $S$ at time $t$} is $\weight^t(S) = |\pop^t(S)|$. 
For $t=0$, we extend the definition by $\pop^0(S) = \emptyset$ and $\weight^0(S) = 0$.  
The \emph{density of $S$ at time $t$} is $\rho^t(S) = \weight^t(S) / |S|$. For a positive 
integer $b$, let $\densify^t(S,b)$ be the smallest subsegment $T$ of $S$ of weight exactly 
$\lfloor (\weight^t(S) - 2b) / 2 \rfloor$ such that $\pop^t(T)$ does not contain any of the $b$ largest 
and smallest elements of $\pop^t(S)$. Hence, $\densify^t(S,b)$ is the densest subsegment of $S$ that contains 
the appropriate number of items but which is surrounded by a large population of $S$ on either side.
If $\pop^t(S)=\{x_1<x_2 <\cdots <x_\ell\}$ then $\midp^t(S)= \lceil (x_{\lceil (\ell-1)/2 \rceil} + x_{\lceil (\ell+1)/2 \rceil} )/ 2 \rceil$.

Let $y^1,y^2,\dots,y^t$ be the first $t$ keys inserted  and let $Rel^t$ be the keys 
that are relabeled by 
$\A$ in response to the insertion of $y^t$. The \emph{busy segment $B^t\subseteq [m]$ at time $t$} is
the smallest segment that contains $f^t(Rel^t)\cup f^{t-1}(Rel^t \setminus \{y^t\})$. 
We say that the algorithm $\A$ is \emph{lazy} if
all the keys that are mapped to $B^t$ are relocated at step $t$, i.e., $(f^t)^{-1}(B^t) = Rel^t$. 
By Proposition~4 in \cite{BKS}, when bounding the cost of $\A$ from below 
we may assume that $\A$ is lazy.

\medskip\noindent {\bf Adversary($\A,n,m,r$)}

\smallskip\noindent Set $p^0=0$.

\smallskip\noindent For $t=1,2,\dots,n$ do
\begin{itemize}
\item
If $t < 8$, let $S^t_1 = [m]$, $b^t_1=1$, $\depth(t)=1$, and $p^{t}=1$.   
Set $y^t = 1 + \lceil (t-1) \cdot (r-1)/6 \rceil$, and run $\A$ on $y^1,y^2,\dots,y^t$ to get $f^t$.
Continue to next $t$.

\item
\emph{Preservation Rule:}
For $i=1,\dots,p^{t-1}$, let $S^t_i=S^{t-1}_{i}$ and $b^t_i=b^{t-1}_i$ . 
(For $t\ge 8$, copy the corresponding segments from the previous time step.)

\item 
\emph{Rebuilding Rule:}

Set $i=p^{t-1}+1$.

While $\weight^{t-1}(S^t_{i-1}) \ge 8$

\begin{itemize}
\item Set $S^t_i = \densify^{t-1}(S^t_{i-1},b^t_{i-1})$.
\item Set $b^t_i = \lceil \weight^{t-1}(S^t_i)/8 \rceil$.
\item Increase $i$ by one.
\end{itemize}

\item
Set $\depth(t)=i-1$.

\item 
Set $y^t = \midp^{t-1}(S^t_{\depth(t)})$.

\item
\emph{The critical level:}
Run $\A$ on $y^1,y^2,\dots,y^t$ to get $f^t$. Calculate $Rel^t$ and $B^t$.
Set the critical level $p^{t}$ to be the largest integer $j\in [\depth(t)]$ such that $B^{t} \subseteq S^{t}_j$.

\end{itemize}

%\smallskip
\noindent\emph{Output:} $y^1,y^2,\dots,y^n$.

\medskip
We make the following claim about the adversary which implies Theorem \ref{thm:main}. We did not attempt to optimize
the constants.

\begin{lemma}\label{l-main}
Let $m,n,r$ be integers such that and $2^{32} \le n \le m \le 2^{\sqrt[4]{n}/8}$ and $2^{n+4}<r$. 
Let $\A$ be a lazy online labeling algorithm with the range $m$. 
Let $y^1,y^2,\dots,y^n$ be the output of {\bf Adversary($\A,n,m,r$)}. 
Then the cost $$\chi_\A(y^1,y^2,\dots,y^n) \ge \frac{1}{512} \cdot \frac{n \log n}{ 3 + \log \lceil \log m \rceil - \log \log n} - \frac{n}{6}.$$
\end{lemma}

%
% MK: if one were to set the bucketing k = \lfloor log m \lfloor then we could
%     drop the ceiling around log m provided that k = \lfloor log m \floor > log n.
%     This comes from a requirement in Lemma 17 for the bucketing lower bound.
%     This would be true for example in m > 10n.

Notice, if $r>2^{n+4}$ then for any $t\in[n-1]$ the smallest pair-wise difference between integers $y^1,y^2,\dots,y^t$ is
at least $2^{n+1-t}$ so $y^{t+1}$ chosen by the adversary is different from all the previous $y$'s. All our analysis will assume
this.
 
To prove the lemma we will design a so called \emph{prefix bucketing game} from the interaction between the
adversary and the algorithm, we will relate the cost of the prefix bucketing to the cost $\chi_\A(y^1,y^2,\dots,y^n)$,
and we will lower bound the cost of the prefix bucketing.

In preparation for this, we prove several useful properties of the adversary.
 
\begin{lemma}
For any $t\in [n]$, $\depth(t)\le \log m$.
\end{lemma}

\begin{proof}
%%MIKE 6-27-2012
The lemma is immediate for $t <8$.
For $t\ge 8$, it suffices to show that the hierarchy $S^t_1,S^t_2,\ldots$ satisfies that for each $i \in [1,\depth(t)-1]$,
$|S^t_{i+1}| \leq |S^t_{i}|/2$.  Recall that $S^t_i$ is obtained from $S^t_{i-1}$ by removing the left and right
buffer to get a subsegment $S'$ and then taking $S^t_{i+1}$ to be the shortest subsegment of $S'$ that contains
exactly half of the keys (rounded down) labeled in $S'$.  Letting $L'$ be the smallest $\lfloor |S'|/2 \rfloor$
labels of $S'$ and $R'$ be the largest $\lfloor |S'|/2 \rfloor$ labels of $S'$, one of $L'$ and $R'$ contains
at least half of the occupied labels (rounded down) in $S'$, which implies $|S^t_{i+1}| \leq |S'|/2 < |S^t_i|/2$.
\end{proof}

% The next proposition follows from definitions.
%
% \begin{proposition}
% For any $t\in[n]$, $i\in[\depth(t)-1]$, and any time $t'\in [\starts(t,i), \sends(t,i)-1]$, 
% $B^{t'} \subseteq S^t_i$. Furthermore, if $\sends(t,i)\le n$ then $B^{\sends(t,i)} \setminus S^t_i \not=\emptyset$.
% \end{proposition}

\begin{lemma}\label{l-p1}
For any $t\in[n]$ and $i\in[\depth(t)-1]$, $64\cdot b^t_{i+1} \ge \weight^{t-1}(S^t_i) - \weight^{t-1}(S^t_{i+1})$.
\end{lemma}

\begin{proof} %% {Lemma \ref{l-p1}}
For $t<8$ the lemma is true trivially so we assume that $t\ge 8$.
For any integers $s,s'$ such that $\starts(t,i) \le s < s' < \sends(t,i)$, 
\begin{align*}
\weight^{s-1}(S^t_i) = \weight^{s'-1}(S^t_i) + (s'-s).
\end{align*}
Let $s=\starts(t,i+1)$. Then $\starts(t,i) \le s \le t <  \sends(t,i+1) \le \sends(t,i)$ so
\begin{align*}
\weight^{t-1}(S^t_i) - \weight^{t-1}(S^t_{i+1}) &= \weight^{s-1}(S^t_i) - \weight^{s-1}(S^t_{i+1}) \cr 
&= \weight^{s-1}(S^s_i) - \weight^{s-1}(S^s_{i+1}).
\end{align*}
Also
\begin{align*}
b^t_{i+1} = b^s_{i+1} = \lceil \weight^{s-1}(S^s_{i+1})/8\rceil \ge  \weight^{s-1}(S^s_{i+1})/8.
\end{align*}
Since $8 \le \weight^{s-1}(S^t_{i})$ and $\weight^{\starts(t,i)-1}(S^t_{i}) \le \weight^{s-1}(S^t_{i}) $
\begin{align*}
b^s_{i} = b^{\starts(t,i)}_{i} &= \lceil \weight^{\starts(t,i)-1}(S^t_{i})/8 \rceil \cr
&\le \weight^{s-1}(S^s_{i})/4.
\end{align*}
Hence,
\begin{align*}
\weight^{s-1}(S^s_{i+1}) = \lfloor  (\weight^{s-1}(S^s_{i}) - 2b^s_i )/ 2 \rfloor &\ge \lfloor \weight^{s-1}(S^s_{i})/4 \rfloor \cr &\ge \weight^{s-1}(S^s_{i})/8.
\end{align*}
Thus, $b^t_{i+1} \ge \weight^{s-1}(S^s_{i})/64 \ge (\weight^{s-1}(S^s_i) - \weight^{s-1}(S^s_{i+1}))/64$. The claim follows.
\end{proof}

\begin{corollary}\label{c-weight}
For any $t\in[n]$ and $i\in[\depth(t)-1]$, $8 + 64 \cdot \sum_{j=i+1}^{\depth(t)} b^t_{j} \ge \weight^{t-1}(S^t_i)$.
\end{corollary}

\begin{lemma}
If $\A$ is lazy then for any $t\in[n]$, $|Rel^t| \ge \sum_{i=p^t+1}^{\depth(t)} b^t_i$.
\end{lemma}

\begin{proof}
If $p^t = \depth(t)$ then the lemma is trivial. If $p^t = \depth(t)-1$ then the lemma is trivial as well
since $|Rel^t|\ge 1$ and $b^t_{\depth(t)}=1$ always. So let us assume that $p^t < \depth(t)-1$.
By the definition of $p^t$ we know that at least one of the following must happen: $f^{t-1}(\min(Rel^t)) \in S^t_{p^t} \setminus S^t_{p^t+1}$, $f^{t}(\min(Rel^t)) \in S^t_{p^t} \setminus S^t_{p^t+1}$, $f^{t-1}(\max(Rel^t)) \in S^t_{p^t} \setminus S^t_{p^t+1}$ or $f^{t}(\max(Rel^t)) \in S^t_{p^t} \setminus S^t_{p^t+1}$. Assume that $f^{t-1}(\min(Rel^t)) \in S^t_{p^t} \setminus S^t_{p^t+1}$ or $f^{t}(\min(Rel^t)) \in S^t_{p^t} \setminus S^t_{p^t+1}$,
the other case is symmetric. Let $Y = \{ y \in \{y^1,y^2,\dots,y^{t-1}\};\; \min(Rel^t) \le y < y^t\}$. 
Since $\A$ is lazy,  $Y \subseteq Rel^t$. $S^t_{p^t+1} \setminus S^t_{\depth(t)}$ is the union of two subsegments, the left one $S_L$ and the right one $S_R$.
The population of $S_L$ at time $t-1$ must be contained in $Y$.
For any $i\in [\depth(t)-1]$, the population of the left subsegment of $S^t_i \setminus S^t_{i+1}$ at time $t-1$ is at least $b^t_i$, by the definition of $S^t_{i+1}$. 
Hence, $\sum_{i=p^t+1}^{\depth(t)-1} b^t_i \le |\pop^{t-1}(S_L)| \le |Y| < |Rel^t|$. Since $b^t_{\depth(t)}=1$, the lemma follows.
\end{proof}

\begin{corollary}\label{c-main}
Let $\A$ be a lazy algorithm. Then $64 \cdot \chi_\A(y^1,y^2,\dots,y^n) + 8n \ge \sum_{t=1}^n \weight^{t-1}(S^t_{p^t})$.
\end{corollary}

\section{Prefix Bucketing}

A prefix bucketing of $n$ items into $k$ buckets is a sequence $a^0,a^1,\dots,a^n\in \N^k$ of bucket configurations
satisfying: $a^0 = (0,0,\dots,0)$ and for $t=1,2,\dots,n$, there exists $p^{t} \in [k]$ such that
\begin{enumerate}
\item $a^t_i = a^{t-1}_i$, for all $i=1,2,\dots,p^{t}-1$, 
\item $a^t_{p^{t}} = 1+ \sum_{i\ge p^{t}} a^{t-1}_i$, and
\item $a^t_i = 0$, for all $i=p^{t}+1,\dots,k$.
\end{enumerate}

The \emph{cost of the bucketing $a^0,a^1,\dots,a^n$} is $c(a^0,a^1,\dots,a^n) = \sum_{t=1}^n a^t_{p^{t}}$. In Section \ref{s-lbb}
we prove the following lemma.

\begin{lemma}\label{l-lbb}
Let $n \ge 2^{32}$ and $k$ be integers where $\log n \le k \le \sqrt[4]{n}/8$.
The cost of any prefix bucketing of $n$ items into $k$ buckets is greater than $\frac{n \log n}{8 (\log 8k - \log \log n)} - n$.
\end{lemma}

We want to relate the cost of online labeling to some prefix bucketing. We will build a specific prefix
bucketing as follows. Set $k=\lceil \log m\rceil$.
For a lazy online labeling algorithm $\A$ and $t=1,\dots,n$, let $f^t,S^t_i,B^t,p^{t},y^t$ and $f^0,p^0$ be as 
defined by the {\bf Adversary($\A,n,m,r$)} and the algorithm $\A$. Denote $Y=\{y^1,y^2,\dots,y^n\}$.
For $t=0,1,\dots,n$ and $i=1,\dots,k$, define a sequence of sets $A^t_i \subseteq Y$ as follows:
for all $i=1,\dots,k$, $A^0_i = \emptyset$, and for $t>0$:
\begin{itemize}
\item $A^t_i = A^{t-1}_i$, for all $i=1,\dots,p^{t}-1$, 
\item $A^t_{p^{t}} = \{y^t\} \cup \bigcup_{i\ge p^{t}} A^{t-1}_i$, and
\item $A^t_i = \emptyset$, for all $i=p^{t}+1,\dots, k$.
\end{itemize} 

The following lemma relates the cost of online labeling to a prefix bucketing.
 
\begin{lemma}\label{l-b2l}
Let the prefix bucketing $a^0,a^1,\dots,a^n$ be defined by $a^t_i = |A^t_i|$, for all $t=0,\dots,n$ and $i=1,\dots,k$.
The cost of the bucketing $a^0,a^1,\dots,a^n$ is at most $64 \cdot \chi_\A(y^1,y^2,\dots,y^n) + 9n$.
\end{lemma}

Lemmas \ref{l-lbb} and \ref{l-b2l} together imply Lemma \ref{l-main}. The following lemma will be used to prove Lemma  \ref{l-b2l}.

\begin{lemma}\label{l-p2}
For any $t\in[n]$ and $i\in [\depth(t)]$,  if $i\not=p^t$ then $f^{t-1}(A^t_i) \subseteq S^t_i$ otherwise
$f^{t-1}(A^t_i \setminus \{y^t\}) \subseteq S^t_i$.
\end{lemma}

\begin{proof}  %% {Lemma \ref{l-p2}}
We prove the claim by induction on $t$. For $t=1$, the only non-empty set is $A^2_1=\{y^1\}$ so the claim is true. 
Let assume that it is true for $t-1$ and we prove it for $t$. 
The {\bf Adversary($\A,n,m,r$)} produces the sets $S^t_i$ and $y^t$, and then the algorithm $\A$ outputs $f^t$.
Based on it the adversary defines $B^t, Rel^t$ and $p^t$.
We distinguish two cases.

\emph{Case $p^{t-1} < p^t$:} For all $i\le p^{t-1}$, $A^{t}_i = A^{t-1}_i$, and for all $i> p^{t-1}$, if $i\not=p^t$ then $A^t_i=\emptyset$
otherwise $A^t_i = \{y^t\}$. For all $i\le p^{t-1}$, $B^{t-1} \subseteq S^{t-1}_{p^{t-1}} \subseteq S^{t-1}_i = S^t_i$, where
the first containment follows from the definition of $p^{t-1}$ and the last equality follows from the definition of $S^t_i$.
For all $i<p^{t-1}$, $y^{t-1} \not\in A^{t-1}_i = A^t_i$ so using the induction hypothesis 
$f^{t-1}(A^t_i) \subseteq f^{t-2}(A^t_i) \cup B^{t-1} \subseteq S^{t-1}_i = S^t_i$.
Similarly for $i=p^{t-1}$, $f^{t-1}(A^t_i) \subseteq f^{t-2}(A^t_i \setminus \{y^{t-1}\}) \cup B^{t-1} \subseteq S^{t-1}_i = S^t_i$.
For each $i>p^{t-1}$, either $A^t_i = \emptyset$ or $i=p^t$ and $A^t_i=\{y^t\}$. In either case the lemma follows trivially.

\emph{Case $p^{t-1} \ge p^t$} is similar: For all $i < p^{t}$, $A^{t}_i = A^{t-1}_i$, for $i=p^t$, $A^t_i = \{y^t\} \cup \bigcup_{j\ge p^t} A^{t-1}_j$,
and for $i>p^t$, $A^t_i=\emptyset$. Again, for all $i\le p^{t}$, $B^{t-1} \subseteq S^{t-1}_{p^{t-1}} \subseteq S^{t-1}_i = S^t_i$.
For all $i<p^{t}$, $y^{t-1} \not\in A^{t-1}_i = A^t_i$ so using the induction hypothesis 
$f^{t-1}(A^t_i) \subseteq f^{t-2}(A^t_i) \cup B^{t-1} \subseteq S^{t-1}_i = S^t_i$.
It only remains to consider the case of $i=p^t$ as the claim is trivial for $i>p^t$.
Since $p^{t-1} \ge p^t$, for $i>p^t$, $S^{t-1}_i \subseteq S^t_{p^t}$.
Using the induction hypothesis, $f^{t-1}(A^t_{p^t} \setminus \{y^t\}) \subseteq f^{t-1}(\bigcup_{j\ge p^t} A^{t-1}_j)
\subseteq \bigcup_{j\ge p^t} f^{t-2}( A^{t-1}_j \setminus \{y^{t-1}\}) \cup B^{t-1} \subseteq S^t_{p^t}$. The lemma follows.
\end{proof}

\begin{proofof}{Lemma \ref{l-b2l}}
Using the previous lemma we see that for $t\in [n]$, $|A^t_{p^t}| \le \weight^{t-1}(S^t_{p^t}) + 1$.
The lemma follows by Corollary \ref{c-main}.
\end{proofof}

\subsection{Lower Bound for Bucketing}\label{s-lbb}

In this section we will prove Lemma \ref{l-lbb}. To do so we will associate with each prefix bucketing 
%% Michal: I find the ''forest'' rather confusing so I cut this out. ''an \emph{ordered rooted forest} consisting of''
a $k$-tuple of ordered rooted trees. 
We prove a lower bound on the sum of depths of the nodes of the trees, and this
will imply a lower bound for the cost of the bucketing.

An {\em ordered rooted tree} is a rooted tree where the children of each node are ordered from left to right.
Since these are the only trees we consider, we refer to them simply as trees.
The \emph{leftmost principal subtree of a  tree $T$} is the subtree rooted in the leftmost child of the root of $T$, the \emph{$i$-th principal subtree of $T$} is the 
tree rooted in the $i$-th child of the root from the left. If the root has less than $i$ children, we consider the $i$-th principal subtree to be empty.
The number of nodes of $T$ is called its {\em size}
and is denoted $|T|$. The {\em depth} of a node is  one more than its distance to the root, i.e., the root has depth 1. The cost
of $T$, denoted $c(T)$, is the sum of depths of its nodes. The cost and size of an empty tree is defined to be zero.

We will be interested in trees that satisfy the following condition.
\begin{definition}[$k$-admissible]
\label{def:k-admissibility}
Let $k$ be a non-negative integer. A tree $T$ is \emph{$k$-admissible} if it contains at most one vertex or
\begin{itemize}
\item its leftmost principal subtree is $k$-admissible and
\item the tree created by removing the leftmost principal subtree from $T$ is $(k - 1)$-admissible.
\end{itemize}
\end{definition}
Notice that when a tree is $k$-admissible it is also $k'$-admissible, for any $k'>k$. The following
easy lemma gives an alternative characterization of $k$-admissible trees: 
\begin{lemma}
\label{lm:k-admissibility}
A rooted ordered tree $T$ is $k$-admissible if and only if for each $i \in [k]$, the $i$-th principal subtree of $T$ is $(k - i + 1)$-admissible.
\end{lemma}
\begin{proof}
%%MIKE 6-27-2012
We show both directions at once by induction on $k$.
For a single-vertex tree the statement holds. Let $T$ be a tree with at least 2 vertices and let
$L$ be its leftmost principal subtree and $T'$ be the subtree of $T$ obtained by removing $L$.
By definition $T$ is $k$-admissible if and only if $L$ is $k$-admissible and $T'$ is
$k-1$ admissible, and by induction on $k$, $T'$ is $k-1$ admissible  if and only if for each $2 \leq i \leq k$
the $i$-th principal subtree of $T$, which is the $(i-1)$-st principal subtree of $T'$
is $(k-i+1)$-admissible.
\end{proof}

We will assign a $k$-tuple of trees $T(\bar{a})_1,T(\bar{a})_2,\dots,T(\bar{a})_k$ to each prefix bucketing $\bar{a} = a^0,a^1,\dots,a^t$. The assignment is
defined inductively as follows. The bucketing $\bar{a} = a^0$ gets assigned the $k$-tuple of empty trees. For bucketing $\bar{a} = a^0,a^1,\dots,a^t$
we assign the trees as follows. Let $p^t$ be as in the definition of prefix bucketing, so $0 \not= a^t_{p^t} \not= a^{t-1}_{p^t}$ and for all $i>p^t$, $a^t_i = 0$.
Let $\bar{a}'=a^0,a^1,\dots,a^{t-1}$.
Then we let $T(\bar{a})_i = T(\bar{a}')_i$, for $1\le i < p^t$, and $T(\bar{a})_i$ be the empty tree, for $p^t < i \le k$. The tree $T(\bar{a})_{p^t}$ consists of
a root node whose children are the non-empty trees among $T(\bar{a}')_{p^t},T(\bar{a}')_{p^t+1},\dots,T(\bar{a}')_k$ ordered left to right by the increasing index.

We make several simple observations about the trees assigned to a bucketing. 

\begin{proposition}\label{p-ts}
For any positive integer $k$, if $\bar{a} = a^0,a^1,\dots,a^t$ is a prefix bucketing into $k$ buckets
then for each $i\in [k]$, $|T(\bar{a})_i| = a^t_i$.
\end{proposition}

The proposition follows by a simple induction on $t$.

\begin{proposition}\label{p-ta}
For any positive integer $k$, if $\bar{a} = a^0,a^1,\dots,a^t$ is a prefix bucketing into $k$ buckets then
for each $i\in [k]$, $T(\bar{a})_i$ is $(k+1-i)$-admissible.
\end{proposition}

Again this proposition follows by induction on $t$ and the definition of $k$-admissibility. The next lemma relates
the cost of bucketing to the cost of its associated trees.

\begin{lemma}\label{l-tc}
For any positive integer $k$, if $\bar{a} = a^0,a^1,\dots,a^t$ is a prefix bucketing into $k$ buckets then
 $\sum_{i=1}^k c(T(\bar{a})_i) = c(\bar{a})$.
\end{lemma}

\begin{proof}
By induction on $t$. For $t=0$ the claim is trivial. Assume that the claim is true for $t-1$ and we will prove it for $t$.
Let $\bar{a}'=a^0,a^1,\dots,a^{t-1}$ and $p^t$ be as in the definition of prefix bucketing.
\begin{align*}
c(\bar{a}) = c(\bar{a}') + 1 + \sum_{i=p^t}^k a^{t-1}_i 
%%%% b/c space: %%% &= \sum_{i=1}^k c(T(\bar{a}')_i) + 1 + \sum_{i=p^t}^k a^{t-1}_i \cr
&= \sum_{i=1}^k c(T(\bar{a}')_i) + 1 + \sum_{i=p^t}^k |T(\bar{a}')_i| \cr
&= \sum_{i=1}^{p^t-1} c(T(\bar{a})_i) + 1 + \sum_{i=p^t}^k \left(c(T(\bar{a}')_i + |T(\bar{a}')_i|\right) 
\end{align*}
% where the second equality follows by the induction hypothesis, next one by Proposition \ref{p-ts}, and the last equality follows
% by the definition of $T(\bar{a})_i$, for $i=1,\dots,p^t-1$.
where the first equality follows by the induction hypothesis and by Proposition \ref{p-ts}, and the last equality follows
by the definition of $T(\bar{a})_i$, for $i=1,\dots,p^t-1$.
For $i\ge p^t$, the depth of each node in $T(\bar{a}')_i$ increases by one when it becomes the child of $T(\bar{a})_{p^t}$ hence
\begin{align*}
c(T(\bar{a})_{p^t}) =  1 + \sum_{i=p^t}^k \left(c(T(\bar{a}')_i + |T(\bar{a}')_i|\right).
\end{align*}
For $i > p^t$, $c(T(\bar{a})_i) = 0$ so the lemma follows.
\end{proof}

Now, we lower bound the cost of any ordered rooted tree.

\begin{lemma}

\label{lm:k-d-complete}
Let $k,d \geq 1$ and $T$ be a $k$-admissible tree of depth $d$.
Then $c(T) \geq d \cdot |T| / 2$ and $|T| \leq \binom{k + d - 1}{k}$.
\end{lemma}
\begin{proof}
We prove the lemma by induction on $k + d$. Assume first that $k = 1$ and $d \geq 1$.
The only $1$-admissible tree $T$ of depth $d$ is a path of $d$ vertices. Hence,
$|T| = d = \binom{1 + d - 1}{1}$ and $c(T)=\sum_{i = 1}^d i  = d \cdot (d + 1) / 2$.

Now assume that $k > 1$ and that $T$ is $k$-admissible.
Denote the leftmost principal subtree of $T$ by $L$ and the tree created by removing $L$ from $T$ by $R$.
By the induction hypothesis and definition of $k$-admissibility it follows that $|T| = |L| + |R| \leq \binom{k + d - 2}{k} + \binom{k + d - 2}{k - 1} = \binom{k + d - 1}{k}$.
Furthermore, $c(T) = c(L) + |L| +  c(R) \geq ((d-1) \cdot |L|  / 2) + |L| + (d \cdot (|T| - |L|) / 2) \geq d \cdot |T| / 2$.
\end{proof}

\begin{lemma}
\label{lm:lower_bound_d} 
Let $n,k$ and $d$ be integers such that $2^{32} \leq n$, $\log n \le k \le \sqrt[4]{n}/8$, and 
$d \le \frac{\log n}{4 (\log 8k - \log \log n)}$. Then ${k+d \choose k} < n$. 
\end{lemma}

\begin{proof} %% {Lemma \ref{lm:lower_bound_d}}
First notice that the expression ${k+d \choose k}$ is increasing in $d$.
Therefore to prove the lemma it suffices to set $d = \left\lfloor \frac{\log n}{4 (\log 8k - \log \log n)} \right\rfloor$ 
and show that ${k +d \choose k} < n$.  For this particular choice of $d$ it also holds that $d < \log n \leq k$. 

To estimate the binomial coefficient we use the fact that for two integers integers $a, b$ 
such that $0 < b \leq a/2$ it holds that ${a \choose b} \leq 2^{H(b/a) a}$
where $H(x) = - x \log x - (1-x) \log (1-x)$ is the binary entropy function.
For $0<x<1/2$, one can bound $H(x) < -2x \log x$. 

Since ${k+d \choose k}= {k+d \choose d}$ and $d < (k + d)/2$ we conclude that
\begin{align*}
{k+d \choose k} = {k+d \choose d} \leq 2^{H\left(\frac{d}{k + d}\right)(k+d)} < 2^{-2d\log\left(\frac{d}{k + d}\right)}.
\end{align*}

From the assumption $k \leq \sqrt[4]{n}/8$ it follows 
that $$\left\lfloor \frac{\log n}{4 (\log 8k - \log \log n)} \right\rfloor \ge \frac{\log n}{8 (\log 8k - \log \log n)}.$$
By substituting for $d$ we get
\begin{align*}
\log {k + d \choose k}
        & < 2d \log\left(\frac{k + d}{d}\right) \cr
        & < 2d \log\left(\frac{2k}{d}\right) \cr
        & \le 2 \cdot \frac{\log n}{4(\log 8k - \log \log n)} \cdot \log \left(\frac{16k (\log 8k - \log \log n)}{\log n}\right) \cr
        & = \frac{\log n}{2} \cdot \frac{\log 2 + \log 8k - \log \log n + \log (\log 8k - \log \log n)}{\log 8k - \log \log n} \cr
        & \le \frac{\log n}{2} \cdot \left(\frac{1}{3} + 1 + \frac{\log 3}{3}\right) \cr
        & < \log n.
\end{align*}

In the next to last inequality we use the fact that $\log 8k - \log \log n \geq 3$ and that $\frac{\log x}{x}$ is decreasing 
when $x \geq 3$.
\end{proof}

\begin{proofof}{Lemma \ref{l-lbb}}
Consider a prefix bucketing $\bar{a} = a^0,a^1,\dots,a^t$ of $n$ items into $k$ buckets, 
where $\log n \le k \le \sqrt[4]{n}/8$. Let $a'^t = (n,0,0,\dots,0)$ be a $k$-tuple of integers and
let  $\bar{a}' = a^0,a^1,\dots,a^{t-1},a'^t$. Clearly, $\bar{a}'$ is also a prefix bucketing and $c(\bar{a}') \le c(\bar{a}) + n-1$. Hence, it suffices to 
show that $c(\bar{a}') \ge  \frac{n \log n}{8 (\log 8k - \log \log n)}$. Let $T$ be $T(\bar{a}')_1$. By Proposition \ref{p-ts}, $|T|=n$, and by Proposition \ref{p-ta}, $T$ is $k$-admissible.
Furthermore, by Lemma \ref{l-tc},  $c(\bar{a}) = \sum_{i=1}^k c(T(\bar{a})_i) = c(T)$. So we only need to lower bound $c(T)$. By Lemma \ref{lm:k-d-complete}
a $k$-admissible tree has size at most $\binom{k + d - 1}{k}$, where $d$ is its depth. For $d \le  \frac{\log n}{4 (\log 8k - \log \log n)}$,  $\binom{k + d - 1}{k} \le \binom{k + d}{k} < n$
by Lemma \ref{lm:lower_bound_d} so $T$ must be of depth $d > \frac{\log n}{4 (\log 8k - \log \log n)}$.
By Lemma \ref{lm:k-d-complete}, $c(T) \ge \frac{n \log n}{8 (\log 8k - \log \log n)}$. The lemma follows.
\end{proofof}

\subsection*{Acknowledgements}

The authors would like to thank V\'aclav Koubek for interesting and helpful discussions on this subject.

\newcommand{\etalchar}[1]{$^{#1}$}

\end{document}